\documentclass[letterpaper, 10 pt, conference]{ieeeconf}

\IEEEoverridecommandlockouts
\usepackage{amsmath}
\usepackage{bbm}
\usepackage{amssymb}
\usepackage{graphicx}
\usepackage{dsfont}
\usepackage{xcolor}
\usepackage[noadjust]{cite}
\usepackage{hyperref}

\newtheorem{theorem}{Theorem}
\newtheorem{proposition}{Proposition}
\newtheorem{assumption}{Assumption}
\newtheorem{lemma}{Lemma}
\newtheorem{cor}{Corollary}
\newtheorem{remark}{Remark}

\providecommand{\quotes[1]}{``#1"}
\newcommand{\prt}[1]{\left(#1\right)}
\newcommand{\brk}[1]{\left[#1\right]}
\newcommand{\brc}[1]{\left\{#1\right\}}
\newcommand{\abs}[1]{\left|#1\right|}
\newcommand{\norm}[1]{\lVert#1\rVert}
\newcommand{\inProd}[2]{\langle#1,#2\rangle}

\DeclareMathOperator*{\argmin}{arg\,min}

\newcommand{\one}{\mathds 1}
\newcommand{\E}{\mathbb{E}}
\newcommand{\Ep}[1]{\E\big[#1\big]}
\newcommand{\R}{\mathbb R}
\newcommand{\Rno}{\mathbb R_{\ge 0}}
\newcommand{\Sn}{S_n}

\newcommand{\Fab}{\mathcal F_{\alpha,\beta}}
\newcommand{\Reg}{\text{Reg}_T}
\newcommand{\Ben}{\text{Ben}_T}
\newcommand{\Pot}{\text{Pot}_T}
\newcommand{\optx}[1]{x^{*,#1}}
\newcommand{\selfx}[1]{x^{s,#1}}
\newcommand{\estx}[1]{x^{#1}}

\newcommand{\CMnew}[1]{{\color{black}#1}}

\title{\LARGE \bf
Resource allocation in open multi-agent systems:\linebreak
an online optimization analysis
}

\author{Renato Vizuete, Charles Monnoyer de Galland, Julien M. Hendrickx, Paolo Frasca and Elena Panteley
\thanks{Research supported in part by the Agence Nationale de la Recherche (ANR) via grant “Hybrid And Networked Dynamical sYstems” (HANDY), number ANR-18-CE40-0010 and by the “RevealFlight” ARC at UCLouvain, by the \textit{Incentive Grant for Scientific Research (MIS)} \quotes{Learning from Pairwise Data} of the F.R.S.-FNRS. }
\thanks{R.~Vizuete and C.~Monnoyer de Galland equally contributed to this work. R.~Vizuete and E.~Panteley are with Universit\'{e} Paris-Saclay, CNRS, CentraleSup\'{e}lec, Laboratoire des signaux et syst\`{e}mes, 91190, Gif-sur-Yvette, France. R.~Vizuete and P.~Frasca are with Univ.\ Grenoble Alpes, CNRS, Inria, Grenoble INP, GIPSA-lab, F-38000 Grenoble, France. C.~Monnoyer de Galland and J. M. Hendrickx are with the ICTEAM institute, UCLouvain, Louvain-la-Neuve, Belgium. C.~Monnoyer de Galland is a FRIA fellow (F.R.S.-FNRS).  (E-mail adresses: ~renato.vizuete@l2s.centralesupelec.fr;
~charles.monnoyer@uclouvain.be;  ~julien.hendrickx@uclouvain.be; ~paolo.frasca@gipsa-lab.fr; ~elena.panteley@l2s.centralesupelec.fr).}}

\begin{document}

\maketitle
\thispagestyle{empty}
\pagestyle{empty}

\begin{abstract}
The resource allocation problem consists of the optimal distribution of a budget between agents in a group. 
We consider such a problem in the context of open systems, where agents can be replaced at some time instances. 
These replacements lead to variations in both the budget and the total cost function that hinder the overall network's performance. 
For a simple setting, we analyze the performance of the Random Coordinate Descent algorithm (RCD) using tools similar to those commonly used in online optimization. 
In particular, we study the accumulated errors that compare solutions issued from the RCD algorithm and the optimal solution or the non-collaborating selfish strategy and we derive some bounds in expectation for these accumulated errors.
\end{abstract}

\section{Introduction}
We consider the optimal resource allocation problem, where a fixed amount of resource must be distributed among $n$ agents while minimizing some separable cost function $f$ \cite{ibaraki1988resource}. 
Problems of this type can be found in many different fields of research including distributed computer systems \cite{kurose1989microeconomic}, games \cite{liang2017distributed}, smart grids \cite{dai2021distributed}, etc. 
In some specific formulations like actuator networks \cite{teixeira2013distributed} or power systems \cite{yi2016initialization}, each agent $i$ holds a quantity $d_i$ (which we call here the \quotes{demand} of agent $i$), so that the total amount of resource to be distributed is $\sum_{i=1}^nd_i$; the problem can then be written as
\begin{align}
    \label{eq:CDC:RA}
    &\min_{x\in\R^{np}} f(x) = \sum_{i=1}^nf_i(x_i)&
    &\hbox{s.t.}&
    &\sum_{i=1}^nx_i=\sum_{i=1}^nd_i,
\end{align}
where each function $f_i:\R^p\to\R$ is $\alpha$-strongly convex and $\beta$-smooth, and represents the local cost held by agent $i$.

Problems of this type have received a lot of attention in the last years and most of them are related to a possible change in the budget due to a variation in the demand of some of the agents 
\cite{bai2018distributed,wang2022distributed}. However, in these works, only quadratic functions are considered which significantly restrict the set of potential cost functions of the agents and do not correspond to the standard assumptions in the field of convex and smooth optimization \cite{nesterov2018lectures,MAL-050}.

In addition to the possible variations of the budget over time in \eqref{eq:CDC:RA}, the composition of the system may also change during the whole process due to the arrival, departure or replacement of agents at a time-scale comparable to that of the process, giving rise to \textit{open multi-agent systems}. 
Those are motivated by the growing size of the systems that tends to slow down the process as compared to the time-scale of potential changes in the set of agents. 
More generally, systems naturally allowing agents to join and leave are becoming common, such as \textit{e.g.} multi-vehicle systems or with the Plug and Play implementation \cite{OMAS:PnP:DecentralizedMPC:riverso2012,OMAS:PnP:PartisionBasedDistrKalman:farina2018}. 
In the case of \eqref{eq:CDC:RA}, it results in the system size $n_t$, the local cost functions $f_i^t$, and the local demands $d_i^t$ becoming time-varying. 
As a consequence, the instantaneous optimum of \eqref{eq:CDC:RA}, denoted $\optx{t}$, changes with the time as well, preventing usual convergence.

Due to these possible changes in the dimension of the system, most of the related works in the field of open multi-agent systems are focused on the analysis of scalar performance indexes associated with the process, which allow overcoming the problem of time-varying dimensions. For instance, in \cite{hendrickx2017open,monnoyer2020open} the variance is proposed as a metric for the analysis of a pairwise gossip algorithm, while in \cite{vizuete2021noise} the mean squared error is the object of study in randomized interactions. 
Regarding optimization, problems such as \eqref{eq:CDC:RA} typically imply a minimization process on a long period, and hence the cost is expected to be paid on a regular basis.
In such setting, a natural way of measuring the performance of an algorithm is to compute its accumulated error with respect to a given strategy over a finite number of iterations.
Similar metrics occur in the context of online optimization \cite{DO:online-varyingFunctions}, where the objective is to minimize the so-called \emph{regret}, commonly defined as the accumulated error of the estimate $\estx{t}$ with respect to $x^*:=\argmin_x \sum_{t=1}^T f^t(x)$, or sometimes with respect to the time-varying solution of \eqref{eq:CDC:RA} $\optx{t}:=\argmin_x f^t(x)$ such as \textit{e.g.}, in \cite{shahrampour2018Online}.
Other extensions of the regret include the case of time-varying constraints, where a similar metric is used to measure the violation \cite{yi2021regret}.

In this work, we analyze the performance of the Random Coordinate Descent algorithm (RCD) \cite{necoara2013random,monnoyer2022random} to solve \eqref{eq:CDC:RA} in open systems.
We study the loss accumulated by the RCD algorithm with respect to the time-varying optimal solution $\optx{t}$ over a finite number of iterations, and its gain with respect to the \emph{selfish strategy} $\selfx{t}$, which consists in the absence of collaboration between the agents (\textit{i.e.}, $\selfx{t}_i=d_i^t$), by obtaining upper bounds.
Finally, we consider the case of quadratic cost functions, for which tighter results are derived.

\section{Problem formulation}
\label{sec:CDC:Statement}
The set of real numbers is denoted by $\R$ and the set of nonnegative integers by $\mathbb{Z}_{\ge 0}$. For two vectors $x,y\in\R^n$, $\inProd{x}{y} = x^\top y = \sum_{i=1}^n x_iy_i$ denotes the usual Euclidean inner product and $\norm{x} = \sqrt{x^\top x}$ the Euclidean norm. 
The set of $n$-dimensional vectors with nonnegative entries is denoted by $\Rno^n$.
We denote the vector of size $n$ constituted of only zeros by $\mathbf{0}_n$ and of only ones by $\mathds{1}_n$.

\subsection{Open resource allocation problem}
\label{sec:CDC:Statement:ORA}
We consider the problem \eqref{eq:CDC:RA}, where we restrict to nonnegative states $x_i\in\Rno^p$ for the agents, and where the local cost functions satisfy the following assumption.
\begin{assumption}[Local cost function]
\label{ass:CDC:LocalCosts}
    The local cost function $f_i:\R_{\ge 0}^p\to\R_{\ge 0}$ of any agent $i$ is 
    \begin{itemize}
        \item continuously differentiable;
        \item $\alpha$-strongly convex: $f_i(x)-\frac{\alpha}{2}\norm{x}^2$ is convex $\forall x$;
        \item $\beta$-smooth: $\norm{\nabla f_i(x)-\nabla f_i(y)}\le \beta\norm{x-y},\forall x,y$;
        \item satisfies $\argmin\nolimits_{x\in\Rno^p}f_i(x)=\mathbf{0}_p$ and $f_i(\mathbf{0}_p)=0$.
    \end{itemize}
    More generally, we use $\Fab^p$ to denote the set of functions $f:\Rno^p\to\Rno$ satisfying these conditions. 
\end{assumption}

Assumption~\ref{ass:CDC:LocalCosts} means that the cost paid by an agent is always nonnegative, and is zero only when the agent does not contribute at all to any activity, \textit{i.e.} $x_i=\mathbf{0}_p$, which is the minimal value taken by $x_i$.
It follows from Assumption~\ref{ass:CDC:LocalCosts} that the global cost satisfies $f(x) = \sum_{i=1}^n f_i(x_i)\in\Fab^{np}$.

To problem \eqref{eq:CDC:RA} we associate an undirected graph $G = (\mathcal V,\mathcal E)$, so that at random times a pair of agents $(i,j)\in \mathcal E$ is uniformly randomly chosen to interact and exchange information.
Moreover, we assume the system is subject to random instantaneous arrivals and departures of agents in the system, respectively resulting in a new agent joining the system with its own local cost function and demand, or in an agent leaving the system and never coming back, with the possibility of sending a last message to their neighbours.
We also consider replacements, which consist in the simultaneous occurrence of both an arrival and a departure.
Hence, the system size $n_t$, the local cost functions $f_i^t$ and the demands $d_i^t$ evolve with time, and consequently the instantaneous solution of \eqref{eq:CDC:RA} is time-varying as well, and is denoted $\optx{t}$.

\subsection{Simplifying assumptions and reformulation}
\label{sec:CDC:Statement:Assumptions}
For this preliminary work, we restrict to the specific case defined by the following assumptions.

\begin{assumption}[1-D functions]
\label{ass:CDC:1Dfunctions}
    The local cost function of any agent at any time is one-dimensional: $f_i^t:\R_{\ge 0}\to\R_{\ge 0}$.
\end{assumption}

\begin{assumption}[Homogeneous demand]
\label{ass:CDC:HomogeneousDemand}
    The demand associated with any agent $i$ at any time $t$ is $d_i^t=1$.
\end{assumption}

\begin{assumption}
\label{ass:CDC:Complete}
    The graph $G = (\mathcal V,\mathcal E)$ is complete.
\end{assumption}

Moreover, we restrict to the case where the openness of the system is solely characterized by replacements of agents (\textit{i.e.}, the simultaneous occurrence of an arrival and a departure), so that the system size is fixed, and $n_t=n$ for all time $t$.
Hence, the system only evolves at the instantaneous occurrences of either pairwise interactions of agents, resulting in a possible update of their states (denoted $U$) or replacements (denoted $R$).
We call these occurrences \quotes{\emph{events}}, and we can define the set of all the events that can possibly take place in the system, which we call \quotes{\emph{event set}}, as follows:
\begin{equation}
    \label{eq:CDC:Statement:Xi}
    \Xi := R\cup U = \prt{\bigcup\nolimits_{i\in\mathcal V}R_i} \cup \prt{\bigcup\nolimits_{(i,j)\in\mathcal E}U_{ij}},
\end{equation}
where $R_i$ denotes the replacement of agent $i$ and $U_{ij}$ a pairwise interaction between agents $i$ and $j$.
We assume that two distinct events never occur simultaneously, so that the system evolves in a discrete manner, where each time-step $k\in\mathbb{Z}_{\ge 0}$ corresponds to the occurrence of an event $\xi_k\in\Xi$.

\begin{assumption}
    \label{ass:CDC:IndepEvents}
    An event $\xi_k$ is independent of all other events $\xi_j$, $j\neq k$ and of the state of the system $x^t$ until time $k$, so that at each time-step either an update (\textit{i.e.}, an event $U$) happens with fixed probability $p$ or a replacement (\textit{i.e.}, an event $R$) with fixed probability $1-p$.
\end{assumption}

Let $S_n := \brc{x\in\Rno^n:\mathds 1^\top x=n}$ be the feasible set, 
we can now express \eqref{eq:CDC:RA} in our setting under the assumptions of this section:
\begin{align}
    \label{eq:CDC:RA_OS}
    \min_{x\in\Sn} f^t(x) = \sum_{i=1}^nf_i^t(x_i).
\end{align}

\subsection{Performance metrics}
\label{sec:CDC:Statement:Metrics}
Natural indexes for measuring the performance of an algorithm in our setting consist in evaluating its accumulated error over a finite number of iterations with respect to a given strategy.
We define the two following strategies of interest in the context of the resource allocation problem:
\begin{itemize}
    \item \textbf{Perfect collaboration}: at each time instant $t$ the agents know the optimal solution of \eqref{eq:CDC:RA_OS} denoted $\optx{t}$;
    \item \textbf{Selfish players}: the agents do not collaborate to minimize $f^t$, and they operate at their individual desired point so that $\selfx{t}=\mathds1_n$ at all $t$.
\end{itemize}
Hence, for any $T$, the estimate $\estx{t}$ obtained with a well-designed algorithm is expected to satisfy
\begin{equation}
    \label{eq:CDC:Metrics:Strategies}
    \sum_{t=1}^Tf^t(\optx{t})\le \sum_{t=1}^Tf^t(x^t)\le \sum_{t=1}^Tf^t(\selfx{t}).
\end{equation}
The evolution of these strategies, compared with that of a given algorithm, is illustrated in Fig.~\ref{fig:CDC:SingleReal}.

\begin{figure}
    \centering
    \includegraphics[width=0.35\textwidth,clip = true, trim=2.5cm 11cm 2.75cm 11.25cm,keepaspectratio]{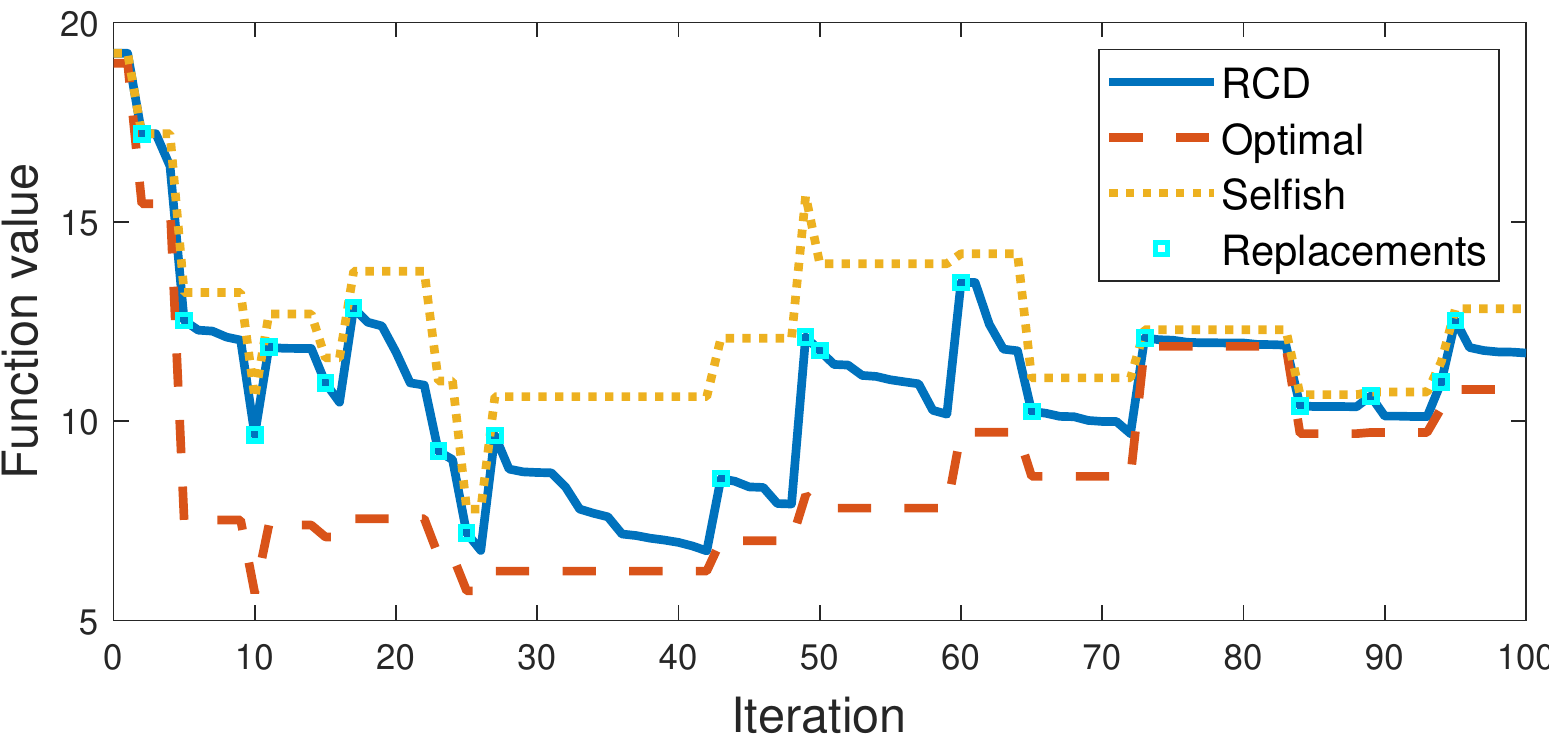}
    \caption{Evolution of the function value $f^t$ evaluated with the RCD algorithm $\estx{t}$ defined in \eqref{eq:CDC:Statement:RCD}, the optimal solution $\optx{t}$, and the selfish strategy $\selfx{t}$, in a system subject to replacements of agents (\textit{i.e.}, simultaneous departures and arrivals) on average once every $4$ RCD steps.}
    \label{fig:CDC:SingleReal}
\end{figure}  

We define the following performance metrics to analyze the value provided by the RCD algorithm $\estx{t}$ with respect to the strategies above: 
\vspace{-0.5cm}
\begin{align}
    \label{eq:CDC:Reg}
    &\hbox{\textit{\textbf{\small Dynamical Regret:}}}&
    &{\small \Reg := \sum_{t=1}^T \prt{f^t(x^t)-f^t(\optx{t})};}\\
    \label{eq:CDC:Ben}
    &\hbox{\textit{\textbf{\small Benefit:}}}&
    &{\small \Ben := \sum_{t=1}^T \prt{f^t(\selfx{t})-f^t(x^t)};}\\ 
    \label{eq:CDC:Pot}
    &\hbox{\textit{\textbf{\small Potential Benefit:}}}&
    &{\small \Pot := \sum_{t=1}^T \prt{f^t(\selfx{t})-f^t(\optx{t})}.}
\end{align}
    
The \quotes{\emph{dynamical regret}} and \quotes{\emph{benefit}} respectively measure the accumulated error from using a given algorithm with respect to the optimal solution $\optx{t}$ and the accumulated gain from using it instead of the selfish strategy $\selfx{t}$.
The \quotes{\emph{potential benefit}} is independent of the algorithm; it represents the accumulated advantage of the optimal strategy with respect to the selfish one, and satisfies $\Pot=\Ben+\Reg$.

Observe that the \emph{regret} commonly used in online optimization typically compares $\estx{t}$ with the overall optimal solution taken over all the iterations, \textit{i.e.}, $x^* = \argmin_{x\in\Sn}\sum_{t=1}^T f^t(x)$.
In that sense, it differs from the \emph{dynamical regret} in \eqref{eq:CDC:Reg}, which compares $\estx{t}$ with the time-varying instantaneous optimal solution $\optx{t} = \argmin_{x\in\Sn} f^t(x)$ at each iteration, such as \textit{e.g.}, in \cite{shahrampour2018Online}.

\subsection{Random Coordinate Descent algorithm and objective}
\label{sec:CDC:Statement:Objective}
We consider the Random Coordinate Descent algorithm (RCD) introduced in \cite{necoara2013random}, such as whenever a pair of agents $(i,j)\in\mathcal E$ interact, they update their respective estimates as
\begin{align}
    \label{eq:CDC:Statement:RCD}
    x_i^+ &= x_i-\tfrac1\beta(f_i'(x_i)-f_j'(x_j))\nonumber\\
    x_j^+ &= x_j-\tfrac1\beta(f_j'(x_j)-f_i'(x_i)).
\end{align}
We moreover assume that whenever an agent $in$ joins the system, it initializes its estimate as
\begin{equation}
    \label{eq:CDC:Statement:ArrivalRule}
    x_{in} = d_{in} = 1,
\end{equation}
and whenever an agent $out$ leaves the system, it sends a last message to all its neighbours (\textit{i.e.}, all the other agents in our setting) with its current estimate $x_{out}$ and its demand $d_{out}$ so the agents $i\neq out$ update their estimates as
\begin{equation}
    \label{eq:CDC:Statement:DepartureRule}
    x_i^+ 
    = x_i + \frac{x_{out}-x_i}{n}
    = \prt{1-\frac1n}x_i+\frac1nx_{out}.
\end{equation}
We show in the following proposition that RCD iterations, arrivals and departures as they are defined in \eqref{eq:CDC:Statement:RCD} to \eqref{eq:CDC:Statement:DepartureRule} guarantee that as long as the initial estimate $x^0$ is feasible, then all the estimates remain feasible. 

\begin{proposition}[Well-posedness]
\label{prop:CDC:Statement:Positivity&RA}
    The event set \eqref{eq:CDC:Statement:Xi} guarantees that if $x^0\in \Sn$, then $x^t\in \Sn$ for all $t$.
\end{proposition}
\begin{proof}
    We first consider arrivals: the nonnegativity of $x_i$ and preservation of the constraint is a direct consequence of  \eqref{eq:CDC:Statement:ArrivalRule}. 
    In the case of departures, the nonnegativity of $x_i$ is a direct consequence of \eqref{eq:CDC:Statement:DepartureRule}. 
    Moreover, if the constraint is satisfied at iteration $k$ with $n_k=n$, then under the departure of the agent labelled $out$ we have
    $$
    \sum_{i\neq out} x_{i}^{k+1}
        =n-x_{out}+\frac{x_{out}}{n}(n-1)-\frac{n-x_{out}}{n} 
        =n-1.
    $$
    We finally consider iterations of the RCD algorithm.
    From Assumptions~\ref{ass:CDC:LocalCosts} and \ref{ass:CDC:1Dfunctions}, it follows that for any $x_i\geq0$
    \begin{align*}
        x_if_i'(x_i)
        \ge \alpha\abs{x_i}^2
        \geq0,
    \end{align*}
    so that $f_i'(x_i)\geq0$.
    Moreover, since $f_i$ is $\beta$-smooth, one has $f_i(x_i)\leq\beta x_i$, and therefore at each update of the RCD algorithm between agents $i$ and $j$ there holds
    $$
    x_i^+=x_i-\frac{1}{2\beta}\prt{f'_i(x_i)-f'_j(x_j)}\ge x_i-\frac{1}{2\beta}(\beta x_i)=\frac{x_i}{2},
    $$
    establishing the nonnegativity of $x_i$. 
    A similar analysis can be used for $x_j$. Due to the symmetry of the update rule, the constraint is always preserved and we conclude the proof.
\end{proof}

Our goal is to analyze the performance of the RCD algorithm \eqref{eq:CDC:Statement:RCD} with the arrival and departure rules \eqref{eq:CDC:Statement:ArrivalRule} and \eqref{eq:CDC:Statement:DepartureRule} in the setting described in Sections~\ref{sec:CDC:Statement:ORA} and \ref{sec:CDC:Statement:Assumptions} using the metrics defined in Section~\ref{sec:CDC:Statement:Metrics} in expectation.

\section{Upper bounds on the performance metrics}
\label{sec:CDC:UpperBounds}

We now derive upper bounds on the evolution of the Potential Benefit and the Dynamical Regret respectively defined in \eqref{eq:CDC:Pot} and \eqref{eq:CDC:Reg} in expectation.
Whereas the former is only related to the problem itself, the latter actually depends on the algorithm we consider.

We first provide the following lemmas, where Lemma~\ref{lem:CDC:UpperBounds:Bounds_||x||} directly follows from the equivalence of the norms.

\begin{lemma}
\label{lem:CDC:UpperBounds:Bounds_||x||}
    Let $x\in \Sn$, then $n\leq\norm{x}^2\leq n^2$. 
\end{lemma}
\begin{lemma}
\label{lem:CDC:UpperBounds:SImple_Bound}
    Let $f(x) = \sum_{i=1}^n f_i(x_i)$, where all $f_i$ satisfy Assumptions~\ref{ass:CDC:LocalCosts} and \ref{ass:CDC:1Dfunctions}, then for any $x\in \Sn$ there holds 
    \begin{align}
        \label{eq:CDC:lem:UpperBounds:SImple_Bound}
        \frac{\alpha}{2}n\le f(x) \leq \frac\beta2n^2.
    \end{align}
\end{lemma}
\vspace{2mm}
\begin{proof}
    From Assumptions~\ref{ass:CDC:LocalCosts} and \ref{ass:CDC:1Dfunctions}, $f(0)=f'(0)=0$.
    Hence, since $f$ is $\beta$-smooth and using Lemma~\ref{lem:CDC:UpperBounds:Bounds_||x||}, there holds $f(x) \leq \frac\beta2\norm{x}^2\leq \frac\beta2n^2$ which establishes the upper bound.
    Similarly, since $f$ is $\alpha$-strongly convex and by using Lemma~\ref{lem:CDC:UpperBounds:Bounds_||x||}, it follows that $f(x) \geq \frac\alpha2\norm{x}^2\geq \frac\alpha2n$, which establishes the lower bound, and concludes the proof.
\end{proof}
Lemma~\ref{lem:CDC:UpperBounds:SImple_Bound} provides a global upper bound on the difference between any two solutions $x,y\in \Sn$:
\begin{align}
    \label{eq:prop:Oversimplified:SimpleBound}
    \abs{f^t(x)-f^t(y)} \leq \frac n2(n\beta-\alpha).
\end{align}
This can be used to derive upper bounds on any of the metrics defined in Section~\ref{sec:CDC:Statement:Metrics}, \textit{e.g.}, $\Ben \leq \frac n2(n\beta-\alpha)T$.

\subsection{Potential Benefit}
\label{sec:CDC:UpperBounds:Pot}
We first obtain in the following theorem an upper bound on the expected value of the potential benefit, which we remind quantifies the accumulated advantage of using the optimal strategy rather than not collaborating at all.
\begin{theorem}
\label{thm:CDC:UpperBounds:PotUB}
    In the setting of Section~\ref{sec:CDC:Statement}, there holds
    \begin{equation}
        \label{eq:thm:CDC:UpperBounds:PotUB}
        \Pot 
        \leq \frac n2\alpha\prt{\kappa-1}T,
    \end{equation}
    and in particular
    \begin{equation}\label{eq:asymptotic_potential_benefit}
        \lim_{T\to\infty} \frac{\Pot}{T} \leq \frac n2\alpha(\kappa-1).
    \end{equation}
\end{theorem}
\begin{proof}
    Remember that $\selfx{t}=\mathds 1_n$ by definition, and that $f^t(\mathbf{0}_n)=0$ and $\nabla f^t(\mathbf{0}_n)=\mathbf{0}_n$ from Assumption~\ref{ass:CDC:LocalCosts}.
    Hence, there holds from the $\beta$-smoothness of $f^t$
    \begin{align*}
        f^t(\selfx{t})
        &\leq \nabla f^t(\mathbf{0}_n)^\top (\selfx{t}) + \frac\beta2\norm{\selfx{t}}^2
        = \frac\beta2\norm{\mathds 1_n}^2
        = \frac\beta2n.
    \end{align*}
    Similarly, since $f^t$ is $\alpha$-strongly convex, we get
    \begin{align*}
        f^t(\optx{t})
        \!&\geq\! \nabla f^t(\mathbf{0}_n)^\top (\optx{t}) + \frac\alpha2\norm{\optx{t}}^2
        = \frac{\alpha}{2}\norm{\optx{t}}^2
        \!\geq\! \frac\alpha2n,
    \end{align*}
    where the last inequality follows from Lemma~\ref{lem:CDC:UpperBounds:Bounds_||x||}.
    Hence
    \begin{align*}
        f^t(\selfx{t})-f^t(\optx{t})
        \leq \frac\beta2n-\frac\alpha2n
        = \frac n2(\beta-\alpha)
    \end{align*}
    holds, and injecting it into \eqref{eq:CDC:Pot} yields \eqref{eq:thm:CDC:UpperBounds:PotUB}.
    The last result then follows from dividing \eqref{eq:thm:CDC:UpperBounds:PotUB} by $T$.
\end{proof}

Notice that since the dynamical regret is nonnegative by definition, the bounds \eqref{eq:thm:CDC:UpperBounds:PotUB} and \eqref{eq:asymptotic_potential_benefit} also hold for the benefit since $\Ben=\Pot-\Reg\leq\Pot$. 

\subsection{Dynamical Regret}
\label{sec:CDC:UpperBounds:Reg}
We now obtain an upper bound on the expected dynamical regret defined in \eqref{eq:CDC:Reg}, where we remind $\estx{t}$ is obtained with the RCD algorithm defined in Section~\ref{sec:CDC:Statement:Objective}.
For that purpose, we first introduce the following intermediate quantities:
\begin{align}   
    \label{eq:CDC:UpperBounds:Ct}
    C_t &:= f^t(x^t)-f^t(\optx{t});\\
    \label{eq:CDC:UpperBounds:Deltaft}
    \Delta f_t &:= f^{t+1}(x^{t+1})-f^t(x^t);\\
    \label{eq:CDC:UpperBounds:Deltaft*}
    \Delta f_t^* &:= f^{t+1}(\optx{t+1})-f^t(\optx{t}).
\end{align}
Thus, $C_t$ corresponds to the instantaneous loss of the RCD algorithm with respect to the optimal solution at iteration $t$, and $\Delta f_t$ and $\Delta f_t^*$ respectively stand for the instantaneous variation at one iteration of the total estimated cost and optimal cost, such that $C_{t+1} = C_t+\Delta f_t-\Delta f_t^*$.

In the following proposition, we study the effect or replacements on $\Delta f_t$ in order to later characterize $C_t$ in expectation, and consequently the expected dynamical regret.

\begin{proposition}
\label{prop:CDC:UpperBounds:Dft:Replacement}
    In the setting of Section~\ref{sec:CDC:Statement} the replacement of an agent, denoted $R$, results in
    \begin{equation}
    \label{eq:prop:CDC:UpperBounds:Dft:Replacement}
        \Ep{\Delta f_t\ |\ R}
        \leq \frac52\beta-\frac{3}{2}\alpha.
    \end{equation}
\end{proposition}
\begin{proof}
    We analyze the effects of arrivals and departures separately.
    Let $g$ denote the local cost function of the joining agent at an arrival, then $f^{t+1}(x^{t+1}) = f^t(x^t)+g(1)$ and
    \begin{align}
    \label{eq:proof:prop:CDC:UpperBounds:Dft:Arrival}
        \Delta f_t 
        = f^t(x^t) + g(1) - f^t(x^t) 
        = g(1)
        \leq \frac\beta2,
    \end{align}
    where the last inequality follows from Assumption~\ref{ass:CDC:LocalCosts}, and in particular the $\beta$-smoothness of $g$.
    
    Consider now a departure, and let $\ell$ denote the label of the leaving agent, such that $f^{t+1}(\estx{t+1}) = \sum_{i\neq \ell} f_i^t(\estx{t+1}_i)$, with $\estx{t+1}_i = \estx{t}_i + \frac{\estx{t}_i+\estx{t}_\ell}{n}$ from \eqref{eq:CDC:Statement:DepartureRule}. 
    From the definition of departures, $\ell$ is uniformly selected among the $n$ agents in the system and by taking the expected value $\Delta f_t$ over the leaving agent, one gets the following, where we omit the reference to time to lighten the notation:
    \begin{align*}
        \E\brk{\Delta f_t}
        &=\sum_{\ell=1}^{n}\frac{1}{n}\prt{\sum_{i\neq \ell}f_i\prt{x_i+\frac{x_\ell-x_i}{n}}-f(x)}\\
        &=\frac{1}{n}\sum_{\ell=1}^{n}\prt{f\prt{\frac{n-1}{n}x+\frac{x_\ell}{n}\one_n}-f_\ell(x_\ell)}-f(x)\\
        &=\frac{1}{n}\sum_{\ell=1}^{n}f\prt{\frac{n-1}{n}x+\frac{x_\ell}{n}\one_n}-\frac{n+1}{n}f(x).
    \end{align*}
    Since $f$ is $\beta$-smooth from Assumption~\ref{ass:CDC:LocalCosts}, one has 
    \begin{multline*}
        f\prt{\frac{n-1}{n}x+\frac{x_\ell}{n}\one_n}\le\\ 
        f(x) + \frac{1}{n}\inProd{\nabla f(x)}{x_\ell\one_n-x}+\frac{\beta}{2n^2}\norm{x_\ell\one_n-x}^2,
    \end{multline*}
    and it follows that
    \begin{align}
        \E\brk{\Delta f_t}
        &\leq \frac{1}{n^2}\sum_{\ell=1}^n\inProd{\nabla f(x)}{x_\ell\one_n-x}\nonumber\\
        &\quad +\frac{\beta}{2n^3}\sum_{\ell=1}^n\norm{x_\ell\one_n-x}^2 - \frac1nf(x).\label{eq:delta_f_t_departure}
    \end{align}
    From Assumption~\ref{ass:CDC:LocalCosts}, in particular since $f$ is $\beta$-smooth and $\alpha$-strongly convex, it satisfies $\alpha\norm{x}^2\le\inProd{\nabla f(x)}{x}\le \beta\norm{x}^2$ for any $x$.
    Hence, reminding that $\sum_{\ell=1}^nx_\ell=n$, the first sum of \eqref{eq:delta_f_t_departure} can be upper bounded by
    \begin{align*}
        \sum_{\ell=1}^n\inProd{\nabla f(x)}{x_\ell\one_n-x}
        &=n\inProd{\nabla f(x)}{\one_n-x}\\
        &\le n(\beta n-\alpha\norm{x}^2).
    \end{align*}
    The second sum of \eqref{eq:delta_f_t_departure} can be expressed as:
    \begin{align*}
        \sum_{\ell=1}^n\norm{x_\ell\one_n-x}^2
        &=\sum_{\ell=1}^n\sum_{i=1}^n\prt{x_\ell^2-2x_\ell x_i+x_i^2}\\
        &=\sum_{\ell=1}^n\prt{nx_\ell^2-2nx_\ell+\norm{x}^2}\\
        &=2n\prt{\norm{x}^2-n}.
    \end{align*}
    Then \eqref{eq:delta_f_t_departure} is upper bounded by
    $$
    \E\brk{\Delta f_t}
    \le \frac{1}{n}\prt{\beta n-\alpha \norm{x}^2} + \frac{\beta}{n^2}\prt{\norm{x}^2-n} -\frac{1}{n}f(x).
    $$
    Lemmas~\ref{lem:CDC:UpperBounds:Bounds_||x||} and \ref{lem:CDC:UpperBounds:SImple_Bound} yield $\norm{x}^2\leq n^2$ and $f(x)\geq \frac\alpha2n$, so that
    \begin{equation}
    \label{eq:proof:CDC:prop:UpperBounds:Dft:Departure}
        \Ep{\Delta f_t}\leq 2\beta-\frac{3}{2}\alpha.
    \end{equation}
    The conclusion follows from adding \eqref{eq:proof:prop:CDC:UpperBounds:Dft:Arrival} and \eqref{eq:proof:CDC:prop:UpperBounds:Dft:Departure}.
\end{proof}

We can now use Proposition~\ref{prop:CDC:UpperBounds:Dft:Replacement} to study the evolution of the expected dynamical regret in the following theorem.

\begin{theorem}
\label{thm:CDC:UpperBounds:RegUB}
    In the setting of Section~\ref{sec:CDC:Statement}, there holds
    \begin{equation}
        \label{eq:thm:CDC:UpperBounds:RegUB}
        \E\Reg
        \leq C_0\sum_{t=1}^T\eta^t + (1-p)\sum_{t=0}^{T-1} \eta^t\prt{M_f+(T-t)\theta},
    \end{equation}
    where $\eta = 1 - \frac{p}{\kappa(n-1)}$ (with $p$ the probability that a given event is an update from Assumption~\ref{ass:CDC:IndepEvents}), $M_f = \frac{n}{2}(\beta n-\alpha)$ and $\theta=\frac52 \beta-\frac{3}{2}\alpha$.
\end{theorem}
\begin{proof}
    Let $\gamma=1-\frac{1}{\kappa(n-1)}$ denote the contraction rate of the RCD algorithm as defined in \eqref{eq:CDC:Statement:RCD} \cite{necoara2013random}.
    Remember that there holds $C_{t+1} = C_t+\Delta f_t + \Delta f_t^*$ for all times $t$.
    Hence, at any time-step $t$ one has 
    \begin{align}
        \label{eq:proof:thm:EReg:Ct+1}
        \Ep{C_{t+1}} 
        &= \Ep{C_t+\Delta f_t-\Delta f_t^*}.
    \end{align}
    From Assumption~\ref{ass:CDC:IndepEvents} the event at iteration $t$ is an update, denoted $U_t$, with probability $p$, or a replacement, denoted $R_t$, with probability $1-p$.
    
\noindent    In the case of an update, we have $\optx{t+1}=\optx{t}$, so that $\Delta f_t^*=0$.
    Hence, we have $\Delta f_t = C_{t+1}-C_t$, and since $\Ep{C_{t+1}|C_t,U_t}\leq \gamma C_t$ with the RCD algorithm from \cite{necoara2013random}:
    \begin{align}
        \label{eq:proof:thm:EReg:UtOnDeltaft}
        \Ep{\Delta f_t|,U_t}
        &= \Ep{C_{t+1}-C_t|U_t}
        \leq (\gamma-1)\Ep{C_t}.
    \end{align}
    In the replacement case, there holds
    \begin{align}
        \label{eq:proof:thm:EReg:RtOnDeltaft}
        \Ep{\Delta f_t|R_t}
        &\leq \theta 
        = \frac52\beta-\frac{3}{2}\alpha,
    \end{align}
    where $\theta$ comes from Proposition~\ref{prop:CDC:UpperBounds:Dft:Replacement}.
    Injecting \eqref{eq:proof:thm:EReg:UtOnDeltaft} and \eqref{eq:proof:thm:EReg:RtOnDeltaft} into \eqref{eq:proof:thm:EReg:Ct+1}
    then yields
    \begin{align}
        \label{eq:proof:thm:EReg:DTsystem}
        \Ep{C_{t+1}}\! 
        &\leq\! \Ep{C_t}\! +\! p(\gamma-1)\Ep{C_t}\! +\! (1-p)\prt{\theta - \Ep{\Delta f_t^*}}\nonumber\\
        &= \eta\Ep{C_t} + (1-p)\prt{\theta - \Ep{\Delta f_t^*}},
    \end{align}
    where $\eta = 1 + p(\gamma-1)$. 
    Expression \eqref{eq:proof:thm:EReg:DTsystem} actually describes the evolution of a discrete-time dynamical system of the type $v[k+1]\leq Av[k]+Bu[k]$.
    Standard results on that framework yield $v[k]\leq A^kv[0]+\sum_{j=0}^{k-1}A^{k-j-1}Bu[j]$, and we obtain
    \begin{align}
        \Ep{C_t} \leq \eta^tC_0 + (1-p) \sum_{j=0}^{t-1}\eta^{t-j-1}\prt{\theta - \Ep{\Delta f_t^*}}.
    \end{align}
    Injecting this last result into \eqref{eq:CDC:Reg} then yields
    \begin{align*}
        &\E\Reg
        = \sum_{t=1}^T \Ep{C_t}\\
        &\ \leq\! C_0\sum_{t=1}^T\eta^t\! +\! (1-p) \sum_{t=1}^T\prt{\sum_{j=0}^{t-1}\eta^{t-j-1}\prt{\theta - \Ep{\Delta f_t^*}}}.
    \end{align*}
    After some term re-organization, it becomes
    \begin{align*}
        \E\Reg
        &\leq C_0\sum_{t=1}^T\eta^t + (1-p)\sum_{t=0}^{T-1}(T-t)\eta^t\theta\\
        &\ \ \ - (1-p)\sum_{t=0}^{T-1}\eta^t\prt{\sum_{j=0}^{T-t}\Ep{\Delta f_j^*}}.
    \end{align*}
    Finally, using Lemma~\ref{lem:CDC:UpperBounds:SImple_Bound}, one concludes that
    \begin{align*}
        &-\sum_{j=0}^{T-t}\Ep{\Delta f_j^*} 
        = -\Ep{\sum_{j=0}^{T-t}\Delta f_j^*}\\
        &\ \ \ = -\Ep{f^{T-t+1}(\optx{T-t+1})-f^1(\optx{1})}
        \leq \frac{n}{2}(\beta n-\alpha),
    \end{align*}
    and $M_f = \frac{n}{2}(\beta n-\alpha)$ yields the conclusion.
\end{proof}

We now analyze the asymptotic behavior of the averaged regret in the following corollary.

\begin{cor}
\label{cor:CDC:UpperBounds:AsymptoticAvgRegUB}
    Let $\rho_R:=\frac{1-p}{p}$.
    In the same setting as that of Theorem~\ref{thm:CDC:UpperBounds:RegUB}, there holds
    \begin{equation}
        \label{eq:cor:CDC:UpperBounds:AsymptoticAvgRegUB}
        \lim_{T\to\infty}\frac{\E\Reg}{T}
        \leq \rho_R(n-1)\beta\prt{\frac{5\kappa-3}{2}}.
    \end{equation}
\end{cor}
\begin{proof}
    Starting from \eqref{eq:thm:CDC:UpperBounds:RegUB}, we have
    \begin{align*}
        \label{eq:thm:Oversimplified:RegretUB:AvgReg}
            \frac{\E\Reg}{T}
            &\leq C_0\sum_{t=1}^T \tfrac{\eta^t}{T} + (1-p)\sum_{t=0}^{T-1}\prt{M_f\tfrac{\eta^t}{T} + \prt{1-\tfrac tT}\eta^t \theta}.
    \end{align*}
    Remember that $\eta = 1-\frac{p}{\kappa(n-1)}\leq1$, so that $\sum_{t=1}^T\eta^t<T$, and $\lim_{T\to\infty}\sum_{t=1}^T \frac{\eta^t}{T}=0$.
    Hence
    \begin{align*}
        \lim_{T\to\infty}\frac{\E\Reg}{T}
        &\leq \lim_{T\to\infty}(1-p)\sum_{t=0}^{T-1}\prt{1-\frac tT}\eta^t \theta.
    \end{align*}
    Moreover, for $\eta<1$, one shows $\sum_{j=0}^\infty \eta^j = \frac{1}{1-\eta}$ and $\sum_{j=0}^\infty j\eta^j = \frac{\eta}{(1-\eta)^2}$.
    The latter implies that $\lim_{T\to\infty}\frac1T\sum_{t=0}^{T-1}t\eta^t=0$, and therefore there holds
    \begin{align*}
        \lim_{T\to\infty}\frac{\E\Reg}{T}
        &\leq (1-p)\frac{1}{1-\eta}\theta
        = \frac{1-p}{p}\kappa(n-1)\theta,
    \end{align*}
    and the conclusion follows from the definitions of $\theta$ in Theorem~\ref{thm:CDC:UpperBounds:RegUB}, and from $\rho_R:=\frac{1-p}{p}$.
\end{proof}

The upper bounds for the potential benefit \eqref{eq:thm:CDC:UpperBounds:PotUB} and the dynamical regret \eqref{eq:thm:CDC:UpperBounds:RegUB} linearly scale with $T$.
This behavior is rather natural for the former, which does not depend on the algorithm.
For the latter, it is most likely unavoidable due to the introduction at each replacement of perturbations of non-decaying magnitude which no algorithm can instantaneously compensate.
Interestingly, this behavior contrasts with standard results in online optimization, where a sublinear growth in $T$ is desired to cancel the asymptotic averaged regret \cite[Ch.~1.1]{DO:online-varyingFunctions}.
However, these results usually apply on another definition of the regret, where $\estx{t}$ is compared with an overall time-independent strategy $x^*$ computed over all $T$ iterations, in opposition with $\optx{t}$ which is optimal for each iteration.

Moreover, the corresponding asymptotic upper bounds linearly grow with $n$ and $\alpha$ for \eqref{eq:asymptotic_potential_benefit}, and with $n-1$ and $\beta$ for \eqref{eq:cor:CDC:UpperBounds:AsymptoticAvgRegUB}, consistently with their expected behavior.
In particular, the \CMnew{scaling of \eqref{eq:cor:CDC:UpperBounds:AsymptoticAvgRegUB} with $n-1$} follows from the convergence rate of the RCD algorithm $\gamma = 1-\frac{1}{\kappa(n-1)}$.
Interestingly, \eqref{eq:cor:CDC:UpperBounds:AsymptoticAvgRegUB} \CMnew{is proportional to} $\rho_R(n-1) = (1-p)\frac{n-1}{p}$, and the bound can thus be seen as the ratio between the probability for a given agent to be involved in a RCD update $\frac{p}{n-1}$ (involved in $\gamma$), and the impact of replacements at the system level $1-p$, independently of $n$.
This is consistent with the bound on the impact of replacements in \eqref{eq:prop:CDC:UpperBounds:Dft:Replacement} which is independent of $n$ (by contrast, alternative situations such as \textit{e.g.}, if all agents were to be reset at each replacement are expected to generate an impact growing with $n$).
Hence, for small values of $\rho_R$ (\textit{i.e.}, rare replacements), the bound guarantees that the asymptotic dynamical regret remains reasonably bounded, and decays to zero when $\rho_R\to0$, \textit{i.e.}, for closed systems.

Finally, observe that \eqref{eq:cor:CDC:UpperBounds:AsymptoticAvgRegUB} \CMnew{is proportional to} $\frac{5\kappa-3}{2}$, consistently with the fact that a larger interval for the possible curvature of the cost functions should generate a larger potential error at replacements.
This factor is a potential source of conservatism, \textit{e.g.}, with respect to \eqref{eq:asymptotic_potential_benefit} where \CMnew{the scaling is in} $\frac12(\kappa-1)$.
More generally, it is not clear yet whether other algorithms than the RCD might provide tighter bounds.

\begin{remark}
The proofs of Theorem~\ref{thm:CDC:UpperBounds:RegUB} and Corollary~\ref{cor:CDC:UpperBounds:AsymptoticAvgRegUB} can directly be adapted to any contraction rate $\gamma<1$, and are thus easily generalized to any other algorithm that guarantees linear convergence; in particular 
$\lim_{T\to\infty}\frac{\E\Reg}{T}\leq \rho_R\frac{\theta}{1-\gamma}$.
\end{remark}

\subsection{The case of quadratic functions}
\label{sec:CDC:UpperBounds:Quad}
The bound on the expected dynamical regret can be refined for the particular case where all local functions are quadratic, \textit{i.e.} satisfy the following additional assumption.
\begin{assumption}[Quadratic functions]
    \label{ass:CDC:UpperBounds:Quad}
    The local cost function of any agent $i$ at time $t$ is of the form
    \begin{align}
        &f_i(x_i) = \phi_ix_i^2,&
        &\phi_i\in\brk{\tfrac\alpha2,\tfrac\beta2}.
    \end{align}
\end{assumption}
\vspace{2mm}
The parameter $\phi_i$ is randomly chosen according to a distribution with a finite support determined by the interval $\brk{\frac{\alpha}{2},\frac{\beta}{2}}$.
Observe that functions satisfying Assumption~\ref{ass:CDC:UpperBounds:Quad} necessarily satisfy Assumption~\ref{ass:CDC:LocalCosts} as well.

Under Assumption~\ref{ass:CDC:UpperBounds:Quad}, we can obtain a tighter bound than that of Proposition~\ref{prop:CDC:UpperBounds:Dft:Replacement}, presented in the following proposition.
\begin{proposition}
\label{prop:CDC:UpperBounds:Quad:Dft:Replacement}
    In the setting of Section~\ref{sec:CDC:Statement}, and under Assumption~\ref{ass:CDC:UpperBounds:Quad}, the replacement of an agent $R$ results in
    \begin{align}
        \label{aq:prop:CDC:UpperBounds:Quad:Dft:Replacement}
        \E\brk{\Delta f_t\ |\ R} 
        \leq \frac{3n^2-3n+1}{2n^2}\prt{\beta-\alpha}.
    \end{align}
\end{proposition}
\vspace{2mm}
\begin{proof}
    The arrival case is treated the same was as in the proof of Proposition~\ref{prop:CDC:UpperBounds:Dft:Replacement}, resulting in $\Delta f_t\leq\frac\beta2$.
    The departure case follows the same first steps with $f_i(x) = \phi_ix^2$, and
    \begin{align*}
        &\E\brk{\Delta f_t}
        = \sum_{\ell=1}^n \frac1n\prt{\sum_{i\neq \ell}\phi_i\prt{x_i+\tfrac{x_\ell-x_i}{n}}^2-f(x)}\\
        &\ \ \ =\frac1n\sum_{\ell=1}^n \prt{\sum_{i\neq \ell} \phi_i\prt{\prt{x_i+\tfrac{x_\ell-x_i}{n}}^2-x_i^2} - \phi_\ell x_\ell^2}\\
        &\ \ \ = \frac{1}{n^2}\sum_{\ell=1}^n\sum_{i\neq l}\phi_i\prt{\tfrac{1-2n}{n}x_i^2+2x_ix_\ell\tfrac{n-1}{n}+\tfrac{x_\ell^2}{n}} - \frac{f(x)}{n}.
    \end{align*}
    Using the fact that $\sum_{\ell=1}^n\sum_{i\neq \ell}\phi_ix_i^2 = (n-1)f(x)$ and that $\phi_i\leq\frac\beta2$ for all $i$, we obtain
    \begin{align*}
        \Ep{\Delta f_t}
        &\leq \prt{\frac{1}{n^2}\frac{1-2n}{n}(n-1)-\frac1n}f(x)\\ 
        &\ \ \ \ + \frac{\beta}{2n^2}\sum_{\ell=1}^n\sum_{i\neq \ell}\prt{2x_ix_\ell\frac{n-1}{n} + \frac{x_\ell^2}{n}}.
    \end{align*}
    Observe that $\sum_{\ell=1}^n \sum_{i\neq l}x_ix_\ell=n^2-\norm{x}^2$ and $\sum_{\ell=1}^n\sum_{i\neq \ell}x_\ell^2 = (n-1)\norm{x}^2$, so that a few algebraic manipulations yield
    \begin{align*}
        \Ep{\Delta f_t}
        &\leq -\frac{3n^2-3n+1}{n^3}f(x) + \beta\frac{n-1}{2n^3}\prt{2n^2-\norm{x}^2}.
    \end{align*}
    Using $\norm{x}^2\geq n$ (from Lemma~\ref{lem:CDC:UpperBounds:Bounds_||x||}) and $f(x)\geq \frac\alpha2n$ (from Lemma~\ref{lem:CDC:UpperBounds:SImple_Bound}) then yields 
    \begin{align*}
    \Ep{\Delta f_t}
        &\leq -\frac{3n^2-3n+1}{2n^3}\alpha + \beta\frac{2n^2-3n+1}{2n^3},
    \end{align*}
    and combining with the arrival case concludes the proof. 
\end{proof}

The result above allows us stating the following theorem, which improves Theorem~\ref{thm:CDC:UpperBounds:RegUB} and Corollary~\ref{cor:CDC:UpperBounds:AsymptoticAvgRegUB} respectively for the case of quadratic functions.
\begin{theorem}
\label{thm:CDC:UpperBounds:Quad:UBReg}
    In the setting of Section~\ref{sec:CDC:Statement}, and under Assumption~\ref{ass:CDC:UpperBounds:Quad}, there holds
    \begin{equation}
        \label{eq:thm:CDC:UpperBounds:Quad:UBReg}
        \E\Reg
        \leq C_0\sum_{t=1}^T\eta^t + (1-p)\sum_{t=0}^{T-1} \eta^t\prt{M_f+(T-t)\theta},
    \end{equation}
    where $\eta = 1 - \frac{p}{\kappa(n-1)}$ (with $p$ the probability that a given event is an update from Assumption~\ref{ass:CDC:IndepEvents}), $M_f = \frac{n}{2}(\beta n-\alpha)$ and $\theta=(\beta-\alpha)\frac{3n^2-3n+1}{2n^2}$.
    In particular,
    \begin{equation}
        \label{eq:thm:CDC:UpperBounds:Quad:AsymptoticUBReg}
        \lim_{T\to\infty}\frac{\E\Reg}{T}
        \leq \rho_R(n-1)\tfrac{3n^2-3n+1}{2n^2}\beta\prt{\kappa-1}.
    \end{equation}
\end{theorem}
\vspace{2mm}
\begin{proof}
The proof follows the exact same steps as those of Theorem~\ref{thm:CDC:UpperBounds:RegUB} and Corollary~\ref{cor:CDC:UpperBounds:AsymptoticAvgRegUB} where Proposition~\ref{prop:CDC:UpperBounds:Quad:Dft:Replacement} is used instead of Proposition~\ref{prop:CDC:UpperBounds:Dft:Replacement}.
\end{proof}

The upper bound for the quadratic case is qualitatively better than that of the general case; it was derived based on an additional information of the cost function, thus resulting in tighter bounds. 
In this case, the dependence of \eqref{eq:thm:CDC:UpperBounds:Quad:AsymptoticUBReg} is in $\frac32(\kappa-1)$, which is consistent with the result derived for the potential benefit \eqref{eq:asymptotic_potential_benefit}.
In particular, for $n$ becoming large, \eqref{eq:thm:CDC:UpperBounds:Quad:AsymptoticUBReg} and \eqref{eq:cor:CDC:UpperBounds:AsymptoticAvgRegUB} become equivalent up to a constant $\beta$.
Moreover, \eqref{eq:thm:CDC:UpperBounds:Quad:AsymptoticUBReg} becomes $0$ when $\kappa=1$, consistently with the expected behavior of the RCD algorithm for quadratic functions since all the cost functions would then be the same.

\subsection{Numerical Results}
To illustrate the results of Theorems~\ref{thm:CDC:UpperBounds:PotUB} to \ref{thm:CDC:UpperBounds:RegUB}, we consider a system of $5$ agents with $\kappa=10$ and $\rho_R=0.0125$, the latter implies that on average there is one replacement every $80$ events.
We consider two possibilities: \emph{random replacements} (RR) where the local function is randomly uniformly chosen among the set of piecewise quadratic functions satisfying Assumption~\ref{ass:CDC:LocalCosts}, and \emph{adversarial replacements} (AR), where these functions are quadratic functions $\phi_ix^2$, with $\phi_i\in\brc{\frac\alpha2,\frac\beta2}$. 
The AR setting is expected to be less favorable than the RR setting, since replacements might result in the largest change of local functions.
\CMnew{Notice that the bounds \eqref{eq:thm:CDC:UpperBounds:PotUB} and \eqref{eq:thm:CDC:UpperBounds:RegUB} are independent of the distribution from which the local cost function are assigned to the agents when they join the system, so that they hold for any such assignment rule.}

\begin{figure}
    \centering
    \includegraphics[width=0.5\textwidth,clip = true, trim=2cm 11.5cm 1.5cm 11.75cm,keepaspectratio]{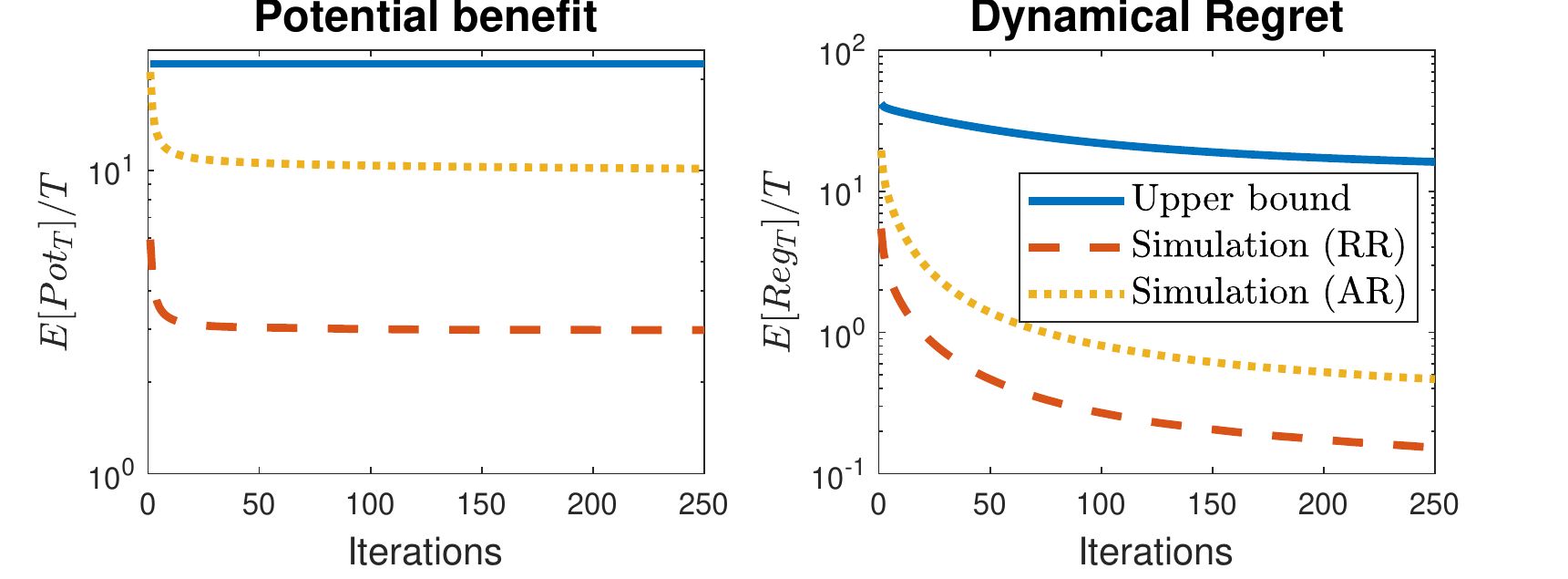}
    \caption{Evolution of the averaged asymptotic expected Potential Benefit (on the left) and dynamical regret (on the right) in a system of $5$ agents with $\rho_R=0.0125$ and $\kappa=10$.
    Each plot compares the upper bounds, respectively from \eqref{eq:asymptotic_potential_benefit} and \eqref{eq:cor:CDC:UpperBounds:AsymptoticAvgRegUB}, with simulated results, either with random replacements (RR) or adversarial replacements (AR).}
    \label{fig:CDC:Illustrations:Comp_Pot_Reg_n5_rho00125_kappa10}
\end{figure}

Fig.~\ref{fig:CDC:Illustrations:Comp_Pot_Reg_n5_rho00125_kappa10} compares the results of Theorem~\ref{thm:CDC:UpperBounds:PotUB} and Corollary~\ref{cor:CDC:UpperBounds:AsymptoticAvgRegUB} with simulations for both random and adversarial replacements in the setting described above. Even though the theoretical bounds are conservative, they capture well the qualitative behavior of these metrics.
In particular, consistently with $\Pot$ and $\E\Reg$ that grow linearly with $T$, the bounds in the figure do not converge to zero, and a remaining asymptotic error is observed.
Our bounds are tighter for the adversarial replacement case than the random replacement case, this suggests that our bounds might be tight for some particular choice of the joining functions at replacements, especially that on the potential benefit.

\begin{figure}
    \centering
    \includegraphics[width=0.4\textwidth,clip = true, trim=2.25cm 11cm 2cm 11cm,keepaspectratio]{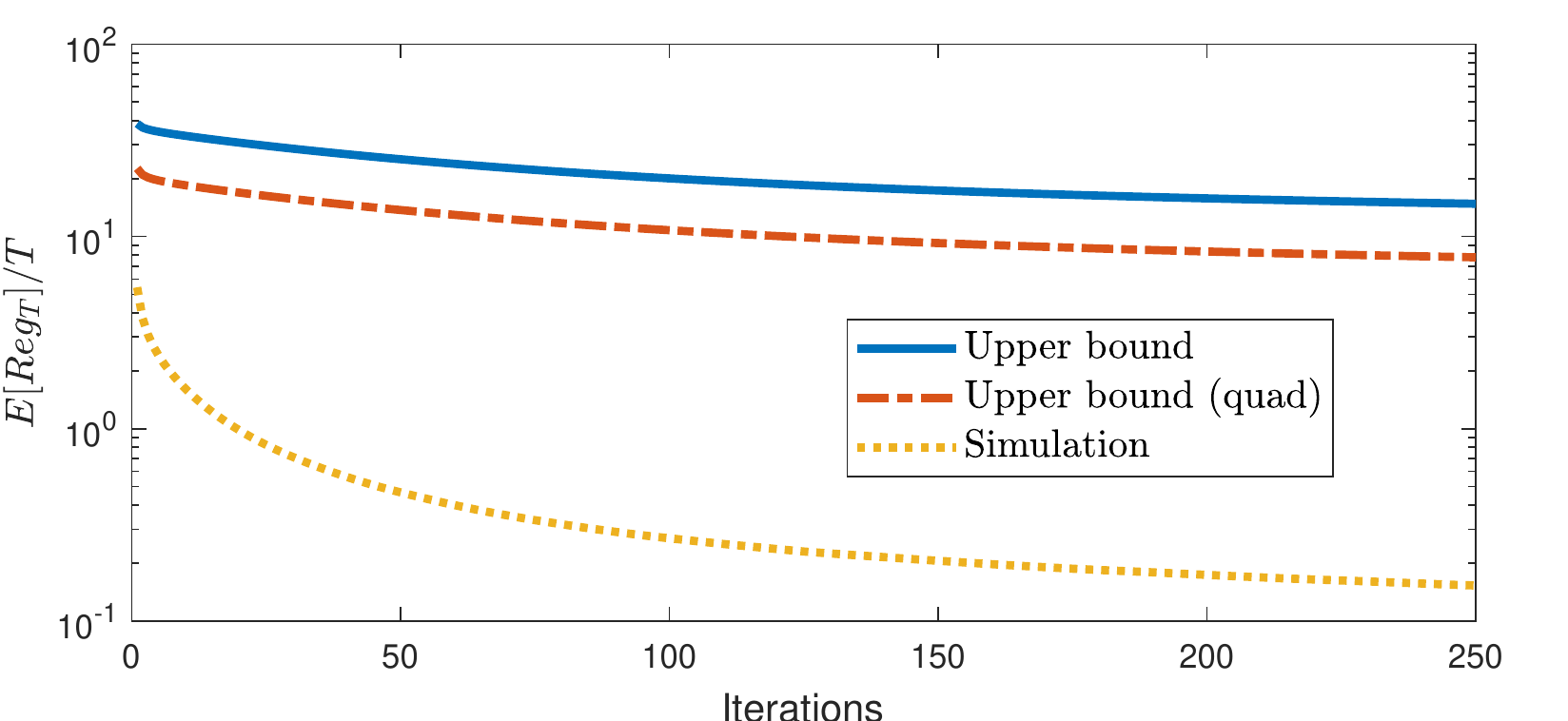}
    \caption{Evolution of the expected averaged regret for a system of $5$ agents holding quadratic functions with $\rho_R=0.0125$ and $\kappa=10$. 
    The plain blue line and the dash-dotted red line respectively correspond to the upper bounds of Theorems~\ref{thm:CDC:UpperBounds:RegUB} and \ref{thm:CDC:UpperBounds:Quad:UBReg} respectively.
    The dotted yellow line corresponds to simulation where we consider random replacements (RR).}
    \label{fig:CDC:Illustrations:RegComp_quad_n5_rho00125_k10}
\end{figure}

Fig.~\ref{fig:CDC:Illustrations:RegComp_quad_n5_rho00125_k10} compares the results of Theorems~\ref{thm:CDC:UpperBounds:RegUB} and \ref{thm:CDC:UpperBounds:Quad:UBReg} with simulations in the same setting as described previously, with random replacements (RR).
The figure shows that bound \eqref{eq:thm:CDC:UpperBounds:Quad:AsymptoticUBReg} is tighter for quadratic functions, following the fact that we have access to more information regarding the local cost functions, thus improving the estimation of the effect of replacements on the expected regret.

\section{Conclusion}
We analyzed the performance and behavior of the Random Coordinate Descent algorithm (RCD) for solving the optimal resource allocation problem in an open system subject to replacements of agents, resulting in variations of the total cost function and of the total amount of resource to be allocated.
We considered a simple preliminary setting where the budget is homogeneous and the graph is complete, and used tools inspired from online optimization to show that it is not possible to achieve convergence to the optimal solution with the RCD algorithm in expectation in open system, but that the error is expected to remain reasonable.

We have derived upper bounds on the evolution of the regret and the potential benefit in expectation and showed that due to the random choice of the new local cost function during replacements, an error is expected to be accumulated with time and cannot be compensated.
A natural continuation of this work is thus the derivation of the corresponding upper bound for the benefit, and of lower bounds for this quantities in order to validate the observed behavior.
More generally, our bounds could be extended to more general settings, and their tightness can be improved to match more accurately the actual performance of the algorithm.
Moreover, since our approach is based on the analysis of the effect of arrivals and departures of agents combined into replacements, the next step of this study is to generalize it to the case where the system size changes with the time, \textit{i.e.}, where arrivals and departures are decoupled.

\bibliographystyle{IEEEtran}

\bibliography{CDC_2022.bib}

\end{document}